\documentclass[preprint,11pt]{elsarticle}
\usepackage[margin=0.95in]{geometry}
\usepackage[noend]{algpseudocode}
\usepackage{algorithm}
\usepackage{times}
\usepackage{amssymb}
\usepackage{amsmath}
\usepackage{latexsym}
\usepackage{graphics}
\usepackage{url}
\usepackage{verbatim}
\usepackage{color}
\usepackage{setspace}
\usepackage{multirow}
\usepackage{ctable}
\usepackage{lineno}
\usepackage{booktabs}
\usepackage{graphicx}

\algnewcommand\And{\textbf{and}}
\algnewcommand\Or{\textbf{or}}
\algnewcommand\Not{\textbf{not}}
\algnewcommand\In{\textbf{in}}
\algnewcommand\Each{\textbf{each}}

\newtheorem{theorem}{Theorem}[section]          
\newtheorem{proposition}[theorem]{Proposition} 
\newtheorem{lemma}[theorem]{Lemma}             
\newtheorem{conjecture}[theorem]{Conjecture}  
\newtheorem{definition}[theorem]{Definition}      

\newenvironment{proof}{{\em Proof.}}{\hspace*{\fill}$\Box$\par\vspace{3mm}}

\newcommand{\squishlist}{
 \begin{list}{$\bullet$}
  { \setlength{\itemsep}{0pt}
     \setlength{\parsep}{3pt}
     \setlength{\topsep}{3pt}
     \setlength{\partopsep}{0pt}
     \setlength{\leftmargin}{2.5em}
     \setlength{\labelwidth}{1em}
     \setlength{\labelsep}{0.5em} } }

\newcommand{\squishlisttwo}{
 \begin{list}{$\triangleright$}
  { \setlength{\itemsep}{0pt}
     \setlength{\parsep}{0pt}
    \setlength{\topsep}{0pt}
    \setlength{\partopsep}{0pt}
    \setlength{\leftmargin}{2em}
    \setlength{\labelwidth}{1.5em}
    \setlength{\labelsep}{0.5em} } }

\newcommand{\squishend}{
  \end{list}  }

\usepackage{listings}

\definecolor{verbgray}{gray}{0.9}

\lstnewenvironment{code}{%
  \lstset{backgroundcolor=\color{verbgray},
 frame=single,
  language=C,
  framerule=0pt,
  basicstyle=\ttfamily,
  showstringspaces=false,
  commentstyle=\color{blue}\textit,
  keywordstyle=\color{black}\bf,
  numbers=none,
  keepspaces=true,
  columns=fullflexible}}{}

\definecolor{shadecolor}{rgb}{.91, .91, .91}
\definecolor{bordercolor}{rgb}{.8, .8, .6}

\definecolor{ultramarine}{rgb}{0, 0.125, 0.376}

 \definecolor{arsenic}{rgb}{0.23, 0.27, 0.29}
 \definecolor{beige}{rgb}{0.96, 0.96, 0.86}

\definecolor{amber}{rgb}{1.0, 0.75, 0.0}
\definecolor{orange}{rgb}{1.0, 0.49, 0.0}
\definecolor{dandelion}{rgb}{0.94, 0.88, 0.19}

  \definecolor{indiagreen}{rgb}{0.07, 0.53, 0.03}
  \definecolor{huntergreen}{rgb}{0.21, 0.37, 0.23}

\newcommand{\blue}[1] {\textcolor{blue}{#1}}

\definecolor{shadecolor}{rgb}{.9, .9, .9}

 \usepackage{framed}

 \colorlet{framecolor}{ultramarine}
  \colorlet{exframecolor}{orange}

    {\endMakeFramed}

    \newenvironment{frshaded*}{%
    \MakeFramed {\advance\hsize-\width \FrameRestore}}%
    {\endMakeFramed}
    
    \newcounter{examplecounter}

{\endMakeFramed}

\newenvironment{frshaded2*}{%
    \MakeFramed {\advance\hsize-\width \FrameRestore}}%
    {\endMakeFramed}

\usepackage{natbib}
\setcitestyle{numbers,sort}
\journal{Theoretical Computer Science}

\newcommand{\Lang}{\mathcal{L}}

\def\swap{\textit{swap}}

\newcommand{\rank}{{\textit{rank}}}
\renewcommand{\epsilon}{\varepsilon}

\newcommand{\PNF}{\mathrm{PNF}}
\newcommand{\LPN}{{\mathcal L}_{\textrm{PN}}}

\newcommand{\pnw}{\textit{pnw}}
\newcommand{\Oh}{{\cal O}}
\newcommand{\Exp}{{\it Exp}}

\begin{document}
\begin{frontmatter}
\title{Generating a Gray code for prefix normal words in amortized polylogarithmic time per word}

\author[elte]{P{\'e}ter Burcsi}%
\ead{bupe@inf.elte.hu}

\author[unipa]{Gabriele Fici}%
\ead{gabriele.fici@unipa.it}

\author[univr]{Zsuzsanna Lipt{\'a}k\corref{cor1}}%
\ead{zsuzsanna.liptak@univr.it}

\author[leic]{Rajeev Raman}%
\ead{r.raman@leicester.ac.uk}

\author[guelph]{Joe Sawada}%
\ead{jsawada@uoguelph.ca}

\address[elte]{Dept. of Computer Algebra, E{\"o}tv{\"o}s Lor{\'a}nd University, Budapest, Hungary}
\address[unipa]{Dip.\ di Matematica e Informatica, Universit{\`a} degli Studi di Palermo, Italy}
\address[univr]{Dip.\ di Informatica, Universit{\`a} degli Studi di Verona, Italy}
\address[leic]{Dept.\ of Informatics, University of Leicester, UK}
\address[guelph]{School of Computer Science, University of Guelph, Canada}

\cortext[cor1]{Corresponding author}

\begin{abstract}
A prefix normal word is a binary word with the property that no substring has more $1$s than the prefix of the same length.  
By proving that the set of prefix normal words is a bubble language, we can exhaustively list all prefix normal words of length $n$ as a combinatorial Gray code, where successive strings differ by at most two swaps or bit flips. This Gray code can be generated in $\Oh(\log^2 n)$ amortized time per word, while the best generation algorithm hitherto has $\Oh(n)$ running time per word. We also present a membership tester for prefix normal words, as well as a novel characterization of bubble languages.
\end{abstract}

\begin{keyword} prefix normal words, binary languages, combinatorial Gray code, combinatorial generation, jumbled pattern matching
\end{keyword}

\end{frontmatter}
\section{Introduction}
\label{sec:intro}

A binary word of length $n$ is  {\em prefix normal} if for all $1 \leq k \leq n$, no substring of length $k$ has more $1$s than the prefix of length $k$. 
For example, the following are the 14 prefix normal words of length $n=5$:
\begin{equation*}\label{eq:LPN5}
11111, 11110, 11101, 11100, 11011, 11010, 11001, 11000, 10101, 10100, 10010, 10001, 10000, 00000.
\end{equation*}

The word $10011$, for instance, is not prefix normal because the substring $11$ has more $1$s than the prefix $10$; similarly $11100110110$ is not prefix normal, because the substring $11011$ has more $1$s than the prefix of length $5$. 
The number $\pnw(n)$ of prefix normal words for $n=1,2,\ldots , 12$ is 
\[ 2, \  3, \  5, \  8,  \ 14, \  23, \   41, \  70,  \ 125,  \ 218,  \ 395,  \mbox{ and } 697,  \]
respectively.   This enumeration sequence is included as sequence A194850 in \emph{The On-Line Encyclopedia of Integer Sequences} (OEIS)~\cite{oeis}, listing $\pnw(n)$ up to $n=50$. It is not difficult to show that $\pnw(n)$ grows exponentially. Some bounds and partial enumeration results were presented in~\cite{BFLRS17}, and it was conjectured there that $\pnw(n) = 2^{n - \Theta(\log^2 n)}$. This conjecture was recently proved by Balister and Gerke in~\cite{BG19}. Finding a closed form formula or generating function for $\pnw(n)$, however, remains an open problem.  

Prefix normal words were originally introduced  in~\cite{FL11} by two of the current authors, in the context of {\em Binary Jumbled Pattern Matching (BJPM)}: Given a binary string $w$ of length $n$, and a pair of non-negative integers $(x,y)$, decide whether $w$ has a substring with $x$ $1$s and $y$ $0$s. While the online version of this problem can be solved naively in $O(n)$ time, the indexed version has attracted much attention during the past decade~\cite{IJFCS12,BurcsiCFL12_ToCS,MoosaR_JDA12,GiaGrab_IPL13,PhTRS-A14,ChanL15,AmirCLL_ICAPL14,CicaleseLWY14,AmirAHLLR16,GagieHLW15,KociumakaRR17,CunhaDGWKS17,ADKN20}. As was shown in~\cite{FL11,BFLRS17}, every binary word $w$ can be assigned two canonical prefix normal words, called its {\em prefix normal forms}, which can then be used to answer BJPM queries in constant time. 

\subsection{Our contributions}

In this paper, we deal with the question of {\em generating} all prefix normal words of a given length $n$. In combinatorial generation, the aim is to exhaustively list all instances of a combinatorial object. Typically, the number of these instances grows exponentially, and time is measured per object, and {\em excluding} the time for outputting the objects. For an introduction to combinatorial generation, see~\cite{Ruskeybook}. 

The current best generation algorithm for prefix normal words runs in $\Oh(n)$ time per word~\cite{CLR18}. Our algorithm improves on this considerably, using amortized $\Oh(\log^2 n)$ time per word. %
It is based on the theory of {\em bubble languages}~\cite{RSW12,SW12,Wi09}, an interesting class of binary languages defined by the following property: ${\cal L}$ is a bubble language if, for every word  in ${\cal L}$, replacing the first occurrence of $01$ (if any)  by $10$ results in another word in ${\cal L}$~\cite{RSW12,SW12}.  Many important languages are bubble languages, including binary necklaces, Lyndon words, and $k$-ary Dyck words\footnote{Those languages are actually $10$-bubble, while  prefix normal words are $01$-bubble: the difference is simply exchanging the role of $0$ and $1$ in the definition, see~\cite{RSW12,SW12}.\label{fn:1}}.  
A generic generation algorithm for bubble languages was given in~\cite{SW12}, yielding  \emph{cool-lex} Gray codes for each  subset of a bubble language containing all strings of a fixed length  and weight (number of 1s). In general, a \emph{(combinatorial) Gray code} is an exhaustive listing of all instances of a combinatorial object such that successive objects in the listing are ``close'' in some well-defined sense.  In the case of \emph{cool-lex} order, the strings differ by at most two swaps.  The generic algorithm's efficiency depends only on a language-dependent subroutine called an \emph{oracle}, which in the best case leads to CAT (constant amortized time) generation  algorithms.

In the following, we show that the set of all prefix normal words forms a bubble language; it is the first new and interesting language shown to be a bubble language since the original exposition~\cite{RSW12, SW12}.  
We develop an oracle for prefix normal words and apply the generic generation algorithm to obtain a cool-lex ordering of prefix normal words with length $n$ and weight $d$.  
Concatenating together these lists  in increasing order of weight, we obtain a Gray code for all prefix normal words of length $n$ where successive words differ by at most two swaps or by a swap and a bit flip.  We then present an  optimized oracle for prefix normal words, and, based on recent results from~\cite{BG19}, we prove that our new generation algorithm runs in amortized time $O(\log^2 n)$ per word. Even though the previous $\Oh(n)$ time per word algorithm of~\cite{CLR18} also provided a Gray code  for prefix normal words (albeit with respect to a different measure of closeness), we are achieving a very considerable improvement in running time. 

As an example, the listing of prefix normal words of length $n=7$ that results from our algorithm, partitioned by weight, 
is given in Table~\ref{table:pnw7}. 
\begin{table}
 \begin{center}
 \begin{tabular}{c@{\hskip 0.2in}c@{\hskip 0.2in}c@{\hskip 0.2in}c@{\hskip 0.2in}c@{\hskip 0.2in}c@{\hskip 0.2in}c@{\hskip 0.2in}c}
 $d=0$ &$d=1$ & $d=2$ & $d=3$ & $d=4$ & $d=5$ & $d=6$ &  $d=7$ \\ \hline
 0000000  &   1000000   	& 1010000 	&  1101000 	&  1101100   	&  1110110   	&  1110111  & 1111111 \\
 		&			& 1001000 	&  1010100 	&  1110100 	&  1111010   	&  1111011  \\
		&			& 1000100 	&  1100100 	&  1101010 	&  1101101   	&  1111101 \\
		&			& 1000010 	&  1010010 	&  1100110 	&  1110101   	&  1111110 \\
		&			& 1000001 	&  1100010 	&  1110010 	&  1101011   	&   \\
		&			& 1100000 	&  1010001 	&  1101001 	&  1110011   	&	\\ 	
		&			&		 	&  1001001 	&  1010101 	&  1111001   	&  \\
		&			&		 	&  1100001 	&  1100101 	&  1111100   	& \\
		&			&		 	&  1110000 	&  1100011 	&  		   	& \\
		&			&		 	&  		 	&  1110001 	&  		   	& \\
		&			&		 	&  		 	&  1111000	&		   	& \\
\end{tabular}
 \end{center}
 \caption{All prefix normal words of length $7$ as output by our algorithm.\label{table:pnw7}}
 \end{table}
 
\bigskip

A second contribution of this paper is a {\em new characterization of bubble languages}. We show that bubble languages can be described in terms of a closure property in the computation tree of a simple recursive generation algorithm for {\em all} binary strings. We believe that this view could aid other researchers in applying the powerful tool of bubble languages and their accompanying Gray codes. In fact, it was the discovery that prefix normal words formed a bubble language that  led to an efficient generation algorithm and  Gray code for our language.

\bigskip
 
The final part of the paper deals with {\em membership testing}, i.e.\ deciding whether a given binary word is prefix normal. Several quadratic-time membership testers for prefix normal words were given in~\cite{FL11,BFLRS17}. The best worst-case time tester can be obtained by using the connection to Indexed Binary Jumbled Pattern Matching (BJPM), for which the current best algorithm, by Chan and Lewenstein, runs in $O(n^{1.864})$ time~\cite{ChanL15}. We present a new membership tester for prefix normal words which applies a simple two-phase approach and is conjectured to run in average-case $O(n)$ time, where the average is taken over all words of length $n$.

\subsection{Related work}

In addition to the connection to jumbled indexing, 
prefix normal words  are also increasingly being studied for their own sake. Enumeration and language-theoretic results were given by Burcsi  et al.\ in~\cite{BFLRS17}, and Balister and Gerke strengthened some results in~\cite{BG19}: in particular, they proved a conjecture about the asymptotic growth behaviour of the number of prefix normal words and gave a new result about the maximal size of the equivalence classes. Cicalese et al.\ gave a generation algorithm in~\cite{CLR18}, with linear running time per word, and studied infinite prefix normal words in~\cite{CLR19}. Prefix normal words and prefix normal forms have been applied to a certain family of graphs by Blondin-Mass\'e et al.~\cite{BM18}, and were shown to pertain to a new class of languages connected to the Reflected Binary Gray Code by Sawada et al.~\cite{SWW17}. Very recently, Fleischmann et al.\ presented some results on the size of the equivalence classes~\cite{FKNP20}.

\subsection{Overview}

The paper is organized as follows. In Section~\ref{sec:preliminaries}, we give the necessary terminology and some basic facts about prefix normal words, and we develop a result on the average critical prefix length of a prefix normal word. This result will later be used in the analysis of our generation algorithm. In Section~\ref{sec:simple}, we give a simple generation algorithm which, based on the result of~\cite{BG19}, is proved to run in amortized $\Oh(n)$ per word. In Section~\ref{sec:bubble}, we present our novel view of bubble languages. In Section~\ref{sec:gray}, we introduce our new generation algorithm, which uses the bubble framework. In Section~\ref{sec:member}, we present the new membership tester. We close with some open problems in Section~\ref{sec:conclusion}.

\bigskip

Some of the results contained in this paper were  presented in a preliminary form at CPM 2014~\cite{BFLRS_CPM14} and FUN 2014~\cite{BFLRS_FUN14}. In particular, the generation algorithm of Sec.~\ref{sec:gray} was originally presented in~\cite{BFLRS_CPM14}, where we proved that it ran in amortized $\Oh(n)$ time per word, and conjectured amortized $\Theta(\log n)$ time per word. Based on the result of~\cite{BG19} on the asymptotic number of prefix normal words, we have been able to prove the amortized $\Oh(\log^2 n)$ running time per word. 

\section{Preliminaries}
\label{sec:preliminaries}

A {\em binary word} (or {\em string}) $w=w_1\cdots w_n$ over $\Sigma=\{0,1\}$ is a finite sequence of elements from $\Sigma$. Its length $n$ is denoted by $|w|$, and the $i$-th symbol of a word $w$ by $w_{i}$, for  $1\leq i\leq |w|$. 
We denote by $\Sigma^n$ the set of words over $\Sigma$ of length $n$, by $\Sigma^{*} = \cup_{n\geq 0} \Sigma^n$ the set of all words over $\Sigma$, and by $\epsilon$ the empty word. 
Let $w\in \Sigma^{*}$. If $w=uv$ for some $u,v\in\Sigma^{*}$, we say that $u$ is a \emph{prefix} of $w$ and $v$ is a \emph{suffix} of $w$. A \emph{substring} of $w$ is a prefix of a suffix of $w$. A {\em binary language} is any subset $\cal L$ of $\Sigma^*$.    
We denote by $|w|_c$ the number of occurrences in $w$ of character $c\in\{0,1\}$. The number of $1$s in $w$, $|w|_1$, is also called the {\em weight} of $w$. For a binary language $\Lang$, let $\Lang(n)$ denote the subset of all strings in $\Lang$ with length $n$, and $\Lang(n,d)$ that of all strings in $\Lang$ with length $n$ and weight $d$. 

We denote by $\swap(w,i,j)$ the string obtained from $w$ by exchanging the characters in positions
$i$ and $j$. 

We define combinatorial Gray codes, following~\cite[ch.\ 5]{Ruskeybook}:  
Given a set of combinatorial objects ${\cal S}$ and a relation $C$ on ${\cal S}$ (the closeness relation), a {\em combinatorial Gray code} for ${\cal S}$ is a listing $s_1, s_2, . . . , s_{|{\cal S}|}$ of the elements of ${\cal S}$, such that $(s_i,s_{i+1}) \in C$ for $i = 1,2,...,|{\cal S}|-1$. If we also require that $(s_{|{\cal S}|},s_1) \in C$, then the code is called {\em cyclic}.

\subsection{Prefix normal words}

Let $w = w_1w_2\cdots w_n$ be a binary word. For each $i=0,1,\ldots,n$, we define  $P(w,i) = |w_1\cdots w_i|_1$, the weight of the prefix of length $i$, 
and $F(w,i) = \max \{|u|_1 : u \text{ is a substring of } w \text{ and } |u|=i\}$, the maximum weight of $i$-length substrings of $w$. 
The function $F$ is sometimes called {\em maximum-ones function}, while in the context of compact data structures, function $P$ is often called $\rank_1(w,i)$~\cite{NavMaek07}.

\begin{definition}
A word $w\in \Sigma^*$ is called {\em prefix normal} if for all $1\leq i \leq |w|$, $F(w,i) = P(w,i)$. We denote by $\LPN$  the language of prefix normal words, and by $\pnw(n) = |\LPN(n)|$, the number of prefix normal words of length $n$. 
\end{definition}

In~\cite{FL11,BFLRS17} it was shown that for every word $w$ there exists a unique word $w'$, called its {\em prefix normal form}, 
such that for all $1\leq i \leq |w|$, $F(w,i) = F(w',i)$, and $w'$ is prefix normal. 
We give the formal definition: 

\begin{definition}
Given a word $w\in \{0,1\}^n$, the {\em prefix normal form of $w$}, denoted $\PNF(w)$, is the prefix normal word $w'$ given by $w'_i = F(w,i)-F(w,i-1)$, for $i=1,\ldots,n$. Two words $w,v$ are {\em prefix normal equivalent} if $\PNF(w) = \PNF(v)$. 
\end{definition}

As an example, the word $w = 11100110110$ has the maximum-ones function $F(w,\cdot) = 0,1,2,3,3,4,4,5,$ $5,6,7,7,$ 
as can be  checked easily. It is furthermore not difficult to see that for all $i<n$: $F(w,i+1) = F(w,i)$ or $F(w,i+1) = F(w,i)+1$. 
Thus the sequence of first differences $F(w,i+1) - F(w,i)$ yields a binary word, in this case the word $11101010110$. 
In  Table~\ref{table:classes5} we list all prefix normal words of length $5$ followed by the set of binary words with this prefix normal form.

\begin{table}[ht]
\begin{small}
\begin{raggedright}
\begin{tabular}{*{2}l @{\hspace{6mm}}||@{\hspace{6mm}} *{2}l}
$\LPN \cap \Sigma^5$ \quad  & Words with this prefix normal form  & $\LPN \cap \Sigma^5$  \quad  & Words with this prefix normal form\\
\hline &&&\rule[-2pt]{0pt}{3pt}\\
$11111$ & \{$11111$\} & $11000$ & \{$11000,01100,00110,00011$\}\\
$11110$ & \{$11110$, $01111$\}& $10101$ & $\{10101\}$\\
$11101$ & \{$11101$, $10111$\}& $10100$ & $\{10100, 01010, 00101\}$\\
$11100$ & \{$11100$, $01110$, $00111$\}& $10010$ & $\{10010, 01001\}$\\
$11011$ & \{$11011$\}& $10001$ & $\{10001\}$\\
$11010$ & \{$11010, 10110, 01101,01011$\}& $10000$ & $\{10000, 01000, 00100, 00010, 00001\}$\\
$11001$ & \{$11001,10011$\}& $00000$ & $\{00000\}$ \\
&&&\\
\hline \vspace{4mm}
\end{tabular}
\end{raggedright}
\caption{All prefix normal words of length 5 and their equivalence classes.\label{table:classes5}}
\end{small}
\end{table}

The next lemma lists some properties of prefix normal words which will be needed in the following. Proofs can be found in~\cite{BFLRS17}. 

\begin{lemma}[\cite{BFLRS17}]\label{lemma:pnw_basics}
Let $w$ be a binary word. 
\begin{enumerate}
\item $w$ is prefix normal if and only if all of its prefixes are prefix normal. 
\item If $w$ is prefix normal, then so is $w0$. 
\item Let $w$ be prefix normal. Then the word $w1$ is prefix normal if and only if, for every suffix $u$ of $w$,  $|u|_1 < P(w,|u|+1)$. 
\item $w$ is prefix normal if and only if $\PNF(w)=w$. 
\end{enumerate}
\end{lemma}

We refer the interested reader to~\cite{BFLRS17} for more on prefix normal words.

\subsection{Critical prefix length}

It will often be useful to write binary words $w\neq 1^n$ as  $w=1^s0^t\gamma$, where $s\geq 0, t\geq1$ and
$\gamma$ is either $\epsilon$ or a binary word beginning with 1. In other words, $s$ is the length of the first, possibly empty, $1$-run of $w$, $t$ is the length of the first $0$-run, and $\gamma$ the remaining, possibly empty, suffix. Note that this representation is unique. 

\begin{definition}
Let $w \in \{0,1\}^n \setminus \{1^n\}$, $w = 1^s0^t\gamma$, where $0\leq s, 1\leq t$ and $\gamma\in 1\{0,1\}^* \cup \{\epsilon\}$. We refer to $1^s0^t$ as $w$'s {\em critical prefix}, and denote by $cr(w) = s+t$ the {\em critical prefix length} of $w$, with $cr(1^n)=n$. 
\end{definition}

For example, the critical prefix length of $11101010110$ is $4$, that of $11111000000$ is $11$, and that of $00101110110$ is $2$.

\begin{lemma}
The expected critical prefix length of a binary string $w$ of length $n$ is $3 - \frac{n+3}{2^n}$. 
\end{lemma}

\begin{proof}
Let $w=1^s0^t\gamma$, with $\gamma \in 1\{0,1\}^* \cup \{\epsilon\}$. Let $X$ be a random variable with $X = cr(w)$, where $w$ is chosen uniformly at random from $\{0, 1\}^n$. 
We have that $X = n$ if and only if $w = 1^s0^{n-s}$, with $s=0,1,\ldots,n$, so for $n+1$ words. Otherwise $X = k$ for some $0<k < n$, and 

$$Pr(X=k) = \sum_{s=0}^{k-1} p^s(1-p)^{k-s}p = \sum_{s=0}^{k-1} (\frac 12)^{k+1} = \frac{k}{2^{k+1}},$$

\noindent since the probability of having a $1$ is $\frac 12$. Therefore,  

$$ \Exp(X) = \sum_{k=0}^{n} k \cdot Pr(X=k) = 
\sum_{k=0}^{n-1} \frac{k^2}{2^{k+1}} + \frac{n(n+1)}{2^n} = 3 - \sum_{k=n}^{\infty} \frac {k^2}{2^{k+1}} + \frac{n(n+1)}{2^n} = 3 - \frac{n+3}{2^n},$$

\noindent where we have used in the last two equations that $\sum_{k\geq 1} \frac{k^2}{2^{k+1}} = 3$, and 
that the tail of the infinite sum has the closed form $2^{-n}(n^2 + 2n +3)$. 
\end{proof}

{\em Remark: } It can be shown in a similar way that the expected critical prefix length of a randomly chosen infinite binary word is $3$.

\begin{lemma}\label{lemma:Crec}
The sequence $C(n)$ of the sum of $cr(w)$ over all words $w$ of length $n$ obeys the recurrence 
$C(n) = 2C(n-1) + (n+1), \text{ with } C(0) = 0.$
\end{lemma}

\begin{proof} Consider all strings of length $n-1$ and what happens to their critical prefix length if one character is added. For those $w$ with $cr(w) < n-1$, it just stays the same, and since we get two new strings $w1$ and $w0$, these $cr(w)$ are counted twice. The remaining strings are either of the form$1^s0^{n-1-s}$, with $s=0,\ldots, n-2$, in which case adding a $0$ will increase $cr(w)$ by $1$, and adding a $1$ will not; there are $n-1$ many of these. Else it is $1^{n-1}$, in which case adding a $0$ or a $1$ will increase $cr(w)$ by $1$. So altogether we get $C(n) = \sum_{|w|=n} cr(w) = \sum_{|w|=n-1} 2cr(w) + (n-1) + 2 = 2C(n-1) + (n+1).$
\end{proof}

Incidentally, the sequence $C(n) = 3\cdot 2^n - (n+3)$, is listed as sequence $A095151$ of the OEIS~\cite{oeis}, along with the second-order recurrence $C(n) = 3C(n-1) - 2C(n-2) + 1,$ for $n \geq 3, \text{ and } C(0)=0, C(1) = 2.$ That it also obeys the recurrence of Lemma~\ref{lemma:Crec} can be seen by computing the difference $C(n+1) - C(n)$ and substituting the recursive formula of order $2$ for both. 

\medskip

Next we show that the expected critical prefix of the prefix normal form of a randomly chosen word is $\Theta(\log n)$. Note that this is not the same as the expected critical prefix length of a random prefix normal word, due to the fact that the equivalence class sizes of the prefix normal equivalence vary considerably (see~\cite{FL11}, and Thm.~2 in~\cite{BG19}).

\begin{lemma}\label{lemma:random_pnf_st}
Given a random word $w$, let $w'$ be the prefix normal form of $w$. Then the expected critical prefix length of $w'$ is $\Theta(\log n)$.
\end{lemma} 

\begin{proof} 
Let $w' = 1^{s'}0^{t'}\gamma'$, with $\gamma'\in 1\{0,1\}^* \cup \{\epsilon\}$ and the r.v.'s $S'=s'$ and $T'=t'$. It is known that the expected maximum length of a run in a random word of length $n$ is $\Theta(\log n)$~\cite{FlajoletSedgewickBook}. Clearly, $S'$ equals the length of the longest run of $1$'s of $w$, thus $\Exp(S') = \Theta(\log n)$. What about $T'$? Consider a $1$-run of $w$ of maximum length $s'$. If $w$ has more than $s'$ $1$'s, then there is a substring of $w$ consisting of this $1$-run and one more $1$; the number of $0$'s in this substring is an upper bound on $t'$. Since these $0$s form one single run, their number is again $O(\log n)$ in expectation. If $w$ has exactly $s'$ $1$'s, then $w' = 1^{s'}0^{n-s'},$ so $t' = n-s' \leq n$. The number of words with at most one $1$-run is ${n+1 \choose 2}+1$. So we have: 
\begin{align*}
Exp(cr(w')) = Exp(S' + T') = \Theta(\log n) + (1 - \frac{\Theta(n^2)}{2^n})O(\log n) + \frac{\Theta(n^2)}{2^n} n = \Theta(\log n).
\end{align*}
\end{proof}

It is not difficult to see that the number of prefix normal words grows exponentially (just note that $1^{|w|}w$ is prefix normal for every $w$). Balister and Gerke~\cite{BG19} recently proved a conjecture from~\cite{BFLRS17} about the asymptotic number of prefix normal words: 

\begin{theorem}[\cite{BG19}, Thm.~1]\label{thm:BG}
The number of prefix normal words of length $n$ is $2^{n-\Theta(\log^2 n)}$. 
\end{theorem}

We will use this theorem to prove an upper bound on $\Exp(cr(w))$ for prefix normal words $w$.  We need the following lemma. 

\begin{lemma}\label{lemma:RR}
Let $t  = o(n)$ and $t\geq 2\log n$, and suppose that $\pnw(n) \geq 2^{n-t}$. Let $Z$ be a random variable taking values $cr(w)$ for prefix normal words $w \in  \LPN(n)$. Then $\Exp(Z) = \Oh(t)$. 
\end{lemma}

\begin{proof}
Consider all prefix normal words $w\in \LPN(n)$ with $cr(w) < 2t$.  These contribute $\Oh(t)$ to $\Exp(Z)$. Now consider all prefix normal words $w\in  \LPN(n)$ with $cr(w) \geq 2t$.  There are at most $(2t + 1) 2^{n - 2t}$ binary words with $cr(w) \geq 2t$, since these words must begin with one of the patterns $1^{2t}, 1^{2t-1}0, 1^{2t-2}00, \ldots , 0^{2t}$, and therefore, at most this number of prefix normal words with $cr(w) \geq 2t$.  Each prefix normal word can only contribute at most $n$ to the average.  So the contribution to the average summed over all prefix normal words with $cr(w) \geq 2t$ is at most $n  (2t + 1)  2^{n - 2t}/\pnw(n) \leq n (2t + 1)  2^{n - 2t}/2^{n -t}$, which is $O(1)$, since $t \geq 2 \log n$, and hence negligible: 
\begin{align*}
\Exp(Z) &= \frac{1}{\pnw(n)} \left(\sum_{\substack{w\in {\LPN (n)}, \\cr(w) < 2t}} cr(w) +  \sum_{\substack{w\in {\LPN (n)}, \\cr(w) \geq 2t}} cr(w)\right) 
\leq \frac{|\{w\in \LPN(n) \mid cr(w) < 2t\}| \cdot 2t}{\pnw(n)} +\\&+ \frac{(2t+1) 2^{n-2t} \cdot n }{2^{n-t}} 
\leq \frac{|\{w\in \LPN(n) \mid cr(w) < 2t\}| \cdot 2t}{\pnw(n)} + \frac{(2t+1)\cdot n}{n^2} 
\leq  2t + \frac{(2t+1)}{n} = \Oh(t). 
\end{align*}

\end{proof}

\begin{theorem}\label{thm:expected-critical-prefix}
The expected length of the critical prefix of a prefix normal word of length $n$ is $\Oh(\log^2 n)$. 
\end{theorem}

\begin{proof} By Theorem~\ref{thm:BG}, we know that there exists a constant $c>0$ such that, for sufficiently large $n$, $\pnw(n) \geq 2^{n-c\log^2 n}$. Applying Lemma~\ref{lemma:RR}, we get $Exp(cr(w)) = \Oh(\log^2n)$, where $w$ ranges over all prefix normal words of length $n$. 
\end{proof}

\section{A Simple Generation Algorithm for Prefix Normal Words}\label{sec:simple}

Our first generation algorithm uses Lemma~\ref{lemma:pnw_basics}: (1) A word is prefix normal if and only if all of its prefixes are prefix normal; (2) if $w$ is prefix normal, so is $w0$, but not necessarily $w1$; and (3) $w1\in \LPN$ if and only if for every suffix $u$ of $w$, the number of ones in $u$ is strictly less than $P(w,|u|+1)$. Words $w\in \LPN$ for which $w1$ is not prefix normal are called {\em extension critical}. Thus, whether a word is extension critical can be tested in linear time in $|w|$. 

We can therefore generate all prefix normal words of length $n$ by iteratively generating all prefix normal words of length $k$, for $k=1,\ldots, n-1$, and extending each one by a 0 if it is  extension critical, or by a 0 and a 1 if it is not. This yields a computation tree whose leaves are precisely the prefix normal words of length $n$. We refer to this algorithm as Simple Generation Algorithm. 

\begin{theorem}
The Simple Generation Algorithm generates all prefix normal words of length $n$ in $O(n)$ amortized time per word. 
\end{theorem}

\begin{proof}
Notice that an extension critical test is performed in each inner node, taking $O(k)$ time if the node is at depth $k$. An inner node at depth $k$ corresponds to a prefix normal word of length $k$, so the number of tests equals the total number of prefix normal words of length smaller than $n$. From Theorem~\ref{thm:BG} (Balister and Gerke, 2019~\cite{BG19}) it follows that most prefix normal words are not extension critical, in particular more than half of all prefix normal words of a given length can be extended. Therefore, $\pnw(n) \geq \frac 32 \pnw(n-1)$, and by induction $\sum_{k=1}^{n-1} \pnw(k) \leq 2\pnw(n)$, implying that the total time taken by the algorithm to generate all prefix normal words of length $n$ is 
\[ \sum_{k=1}^{n-1} k \cdot \pnw(k) \leq \sum_{k=1}^{n-1} n \cdot \pnw(k) = O(n) \cdot \pnw(n).\]

\end{proof}


\section{Bubble Languages and Combinatorial Generation}\label{sec:bubble}

In this section we give a brief introduction to bubble languages, mostly summarising results from~\cite{RSW12,SW12}. However, our presentation is different in that it  presents the generation of a bubble language as a restriction of an algorithm for generating {\em all} binary words. This view also yields a new characterization of bubble languages in terms of the computation tree of this generation algorithm (Prop.~\ref{obs:tree}).

\begin{definition}[\cite{RSW12,SW12}]
A language ${\cal L} \subseteq \{0,1\}^*$ is called {\em a first-$01$ bubble language} if, for every word $w\in {\cal L}$, exchanging the first occurrence of $01$ (if any) by $10$ results in another word in ${\cal L}$. It is called a {\em a first-10 bubble language} if, for every word $w\in {\cal L}$, exchanging the first occurrence of $10$ (if any) by $01$ results in another word in ${\cal L}$. If not further specified, by {\em bubble language} we mean first-01 bubble. 
\end{definition}

For example, the languages of binary Lyndon words and necklaces are $10$-bubble languages. As was shown in~\cite{RSW12}, a language ${\cal L}\subseteq \{0,1\}^n$ is a bubble language if and only if each of its fixed-weight subsets ${\cal L}(n,d)$ is a bubble language. This implies that for generating a bubble language, it suffices to generate its fixed-weight subsets.

\medskip

Next we consider combinatorial generation of binary strings. 
Let $w$ be a binary string of length $n$, let $d$ be its weight, and let $i_1<i_2<\ldots < i_d$ denote the positions of the $1$s in $w$. 
Clearly, we can obtain $w$ from the word $1^d0^{n-d}$ with the following algorithm: first swap the last $1$ with the $0$ in position $i_d$, then swap the $(d-1)$st $1$ with the $0$ in position $i_{d-1}$ etc. Note that every $1$ is moved at most once, and in particular, once the $k$'th $1$ is moved into the position $i_k$, the suffix $w_{i_k}\cdots w_n$ remains fixed for the rest of the algorithm. 

These observations lead to the {\sc Recursive Swap Generation Algorithm} (Algorithm~\ref{algo:bubble}). 
Starting from the string $1^s0^t\gamma$, it generates recursively all $n$-length binary strings with weight $d$ and fixed suffix $\gamma$, where $\gamma\in 1\{0,1\}^* \cup \{\epsilon\}$. 
The call {\sc RecursiveSwap}($d,n-d,\epsilon$) generates all binary strings of length $n$ with weight $d$.  The algorithm swaps the last $1$ of the first $1$-run with each of the $0$s of the first $0$-run, thereby generating a new string each, for which it makes a recursive call. During the execution of the algorithm, the current string resides in a global array $w$.  The function {\sc Swap}($i,j$) swaps the values stored in $w_i$ and $w_j$. In the subroutine {\sc Visit()} we can print the contents of this array, or increment a counter, or check some property of the current string. Crucially, {\sc Visit()} is called on every string exactly once. 

%

\begin{algorithm}[hbt]
\small 
\caption{Recursive Swap Generation Algorithm to generate all binary strings of length $n$.
\label{algo:bubble}}
\begin{algorithmic}[1]

\Statex

\Procedure{RecursiveSwap}{$s,t, \gamma$}

\If {$s>0$ and $t>0$} 
 \For {$i \gets 1$ {\bf to}  $t$} 
   \State \Call{Swap}{$s$,${s{+}i}$}
   \State \Call{RecursiveSwap}{$s{-}1, i, 10^{t-i}\gamma$}
   \State \Call{Swap}{$s$,${s{+}i}$} 
 \EndFor
 \EndIf
\State \Call{Visit}{}()
\EndProcedure

\Statex

\For{$d\gets 0$ {\bf to } $n$}
  \State \Call{RecursiveSwap}{$d,n-d,\epsilon$}
\EndFor

\end{algorithmic}
\end{algorithm}

Let $T^n_d$ denote the computation tree of {\sc RecursiveSwap}($d,n-d,\epsilon$). As an example, Fig.~\ref{fig:example1} illustrates $T^7_4$ (ignore for now the highlighted words). In slight abuse of notation, in the following we identify a node $v$ with the string it represents. 
 The depth of $T^n_d$ equals $d$, the number of $1$s, while the maximum degree (number of children) is $n-d$, the number of $0$s. Consider the subtree rooted at $v = 1^s0^t\gamma$: its depth is $s$ and the maximum degree of nodes is $t$; the number of children of $v$ itself is exactly $t$, and $v$'s $i$th child is $1^{s-1}0^i10^{t-i}\gamma$. Note that suffix $\gamma$ remains unchanged in the entire subtree; that the computation tree is isomorphic to the computation tree of $1^s0^t$; and 
that the critical prefix length strictly decreases along any downward path in the tree. 
The algorithm performs a post-order traversal of the tree, yielding a listing of the strings of length $n$ with weight $d$, in what is referred to as {\em cool-lex order}~\cite{Wi09,SW12,RSW12}.

\begin{figure}
\begin{center}
\includegraphics[width=1\textwidth]{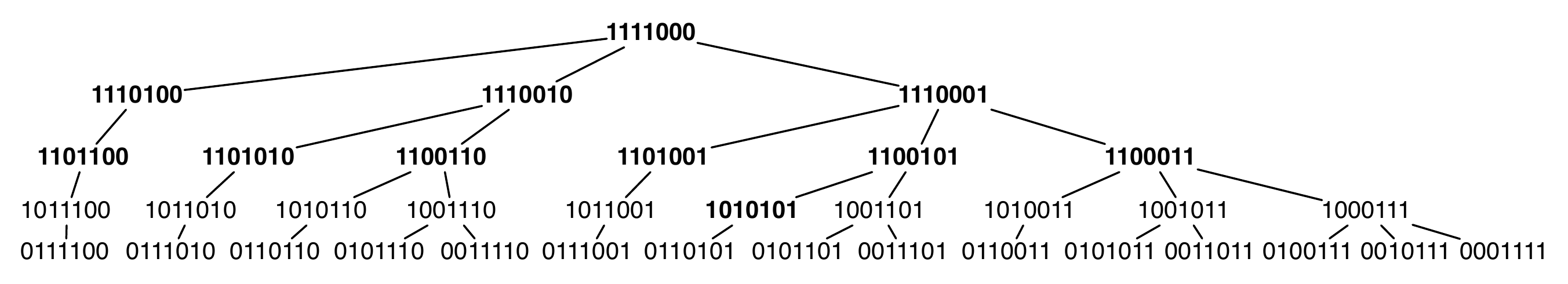}
\caption{\label{fig:example1}The computation tree $T_d^n$ for $n=7,d=4$. Prefix normal words in bold.}%
\end{center}
\end{figure}

We can express the property of bubble language in terms of the computation tree $T_d^n$ as follows:

\begin{proposition}\label{obs:tree}
A language ${\cal L}\subseteq \{0,1\}^n$ is a bubble language if and only if, for every $d=0,\ldots,n$, its fixed-density subset ${\cal L}(n,d)$ is closed w.r.t.\ parents and left siblings in the computation tree $T_d^n$ of the Recursive Swap Generation Algorithm. In particular, if ${\cal L}(n,d) \neq \emptyset$, then it forms a subtree rooted in $1^d0^{n-d}$.
\end{proposition}

\begin{proof}
Follows immediately from the definition of bubble languages. 
\end{proof}

Using  Prop.~\ref{obs:tree}, the {\sc Recursive Swap Generation Algorithm} can be applied to generate {\em any} fixed-weight bubble language ${\cal L}$, as long as we have a way of deciding, for a node $w=1^s0^t\gamma$, already known to be in ${\cal L}$, which is its {\em rightmost child} (if any) that is still in ${\cal L}$. If such a child exists, and it is the $j$th child $u=1^{s-1}0^j10^{t-j}\gamma$, then the bubble property ensures that all children to its left are also in ${\cal L}$. Thus, line $2.$ in Algorithm~\ref{algo:bubble} can simply be replaced by ``for $i=1,\ldots,j$''. 

The framework provided in~\cite{RSW12,SW12} to list the strings in $\Lang(n,d)$ for a given bubble language $\Lang$, can thus be viewed as a restriction of the {\sc Recursive Swap Generation Algorithm}: 
Given a string $w=1^s0^t\gamma \in \Lang$, compute the largest integer $j$ such that $1^{s-1}0^{j}10^{t-j}\gamma \in \Lang$, in other words, the rightmost child of node $1^s0^t\gamma \in \Lang$ which is still in $\Lang$, called the \emph{bubble upper bound}\footnote{In~\cite{RSW12,SW12}, actually a ``bubble lower bound'' is computed.  Because we feel it simplifies the discussion, here we introduce a related value called the ``bubble upper bound''.  The bubble lower bound is equal to $t$ minus the bubble upper bound.}. 
This simple framework is outlined in Algorithm~\ref{algo:framework} for a given bubble language $\Lang$, where the current word $w=w_1w_2\cdots w_n$ is stored globally.
The function {\sc Oracle}($s,t$) returns the bubble upper bound for $w$ with respect to $\Lang$.    The  membership tester {\sc Member}($\Lang, w)$ returns true if and only if $w \in \Lang$.  The initial call is {\sc GenBubble}($d,n-d$) with $w$ initialized to $1^d0^{n-d}$.  

\begin{algorithm}[hbt]
\small
  \caption{Generic algorithm to list $\Lang(n,d)$ for a given bubble language $\Lang$  in cool-lex order.  \label{algo:framework}}

  \begin{algorithmic}[1]

\Statex  
\Function{Oracle}{$s, t$}

	\State $j \gets 1$
	\While {$j \leq t$ {\bf and}  \Call{Member}{$\Lang, 1^{s-1}0^{j}10^{t-j}\gamma$} }  \ $j \gets j+1$  \EndWhile
	\State \Return $j-1$

\EndFunction
 
\Statex

\Procedure{GenBubble}{$s, t$}

\If {$s > 0$ {\bf and} $t > 0$}  
	\For {$i \gets 1$ {\bf to}  \Call{Oracle}{$s,t$}}    
		\State \Call{Swap}{$s$, ${s{+}i}$} 
		\State \Call{GenBubble}{$s{-}1, i$} 
		\State \Call{Swap}{$s$, ${s{+}i}$} 
	\EndFor
\EndIf
\State \Call{Visit}{}()

\EndProcedure

  \end{algorithmic}
\end{algorithm}

\

It was further shown in~\cite{RSW12} that cool-lex order, the order in which the generic algorithm visits the strings of ${\cal L}(n,d)$, gives a Gray code. This can be seen on the tree $T_d^n$ as follows: 

\begin{lemma}\label{lemma:swap}
Let $u$ be a node in the computation tree $T_d^n$. Then each of the following can be obtained from $u$ by a single swap operation: (a) any sibling of $u$, (b) $parent(u)$, and (c) any node on the leftmost path in the subtree rooted in $u$. 
\end{lemma}

\begin{proof}
Let $u$ and $u'$ be siblings, and let their parent be $v = 1^s0^t\gamma$. Then there exist $i,j$ such that $u = 1^{s-1}0^i10^{t-i}\gamma$ and $u'=1^{s-1}0^j10^{t-j}\gamma$. Then $u' = swap(u,s+i,s+j)$, while $v = parent(u) = swap(u,s+i,s)$. For (c), let $u  = 1^s0^t\gamma$ and $u' = 1^k01^{s-k}0^{t-1}\gamma$ for some $k$; then $u' = swap(u,k+1,s+1)$. 
\end{proof}

We now report the main result on bubble languages from~\cite{RSW12,SW12}, for which we give a proof using Prop.~\ref{obs:tree}. 

\begin{proposition}[\cite{RSW12,SW12}] \label{prop:bubble-gray}
Any fixed-length bubble language $\Lang(n)$, where $\Lang(n,d) \neq \emptyset$ for all $d=0,\ldots,n$, can be generated such that subsequent strings differ by at most two swaps, or by a swap and a bit flip. Given a membership tester {\sc Member($\Lang,w$)} which runs in $\Oh(m)$ time, this generation algorithm takes amortized $\Oh(m)$ time per word. 
\end{proposition}

\begin{proof}
For a fixed-weight subset ${\cal L}(n,d)$, 
let $T_{{\cal L}}$ denote the subtree of $T_d^n$ corresponding to ${\cal L}(n,d)$. 
Note that in a post-order traversal of $T_{{\cal L}}$, we have:
\[next(u) = 
\begin{cases}
parent(u) & \text{if } u \text{ is rightmost child}\\
\text{leftmost descendant of $u$'s right sibling} & \text{otherwise.}\\
\end{cases}
\]

By Prop.~\ref{obs:tree}, we have that the leftmost descendant of any node in $T_{{\cal L}}$ lies on the leftmost path in $T_d^n$. Thus, by Lemma~\ref{lemma:swap}, $next(u)$ can be reached in one or two swaps. 

By concatenating the lists for weights $0,1,\ldots , n$, the procedure {\sc GenerateAll}($n$)  shown in Algorithm~\ref{algo:all}  will exhaustively list $\Lang(n)$ for a given bubble language $\Lang$.  To see the Gray code property, notice that for any weight $d$, the last string visited is $1^d0^{n-d}$, while the first string visited for the next weight $d+1$ is the leftmost descendant of $1^{d+1}0^{n-d-1}$, i.e.\ a string of the form $1^i01^{d+1-i}0^{n-d-2}$, which is one swap and one bit flip away from $1^d0^{n-d}$. 

For the running time, notice that for $w = 1^s0^t\gamma$, we do at most $j+1$ membership tests, where $j$ is the bubble upper bound for $w$. The $j$ successful tests can be charged to the $j$ children of $w$, while the possible last unsuccessful test can be charged to $w$ itself. 
\end{proof}

%
\begin{algorithm}[hbt]
  
  \caption{A Gray code to exhaustively list $\Lang(n)$ for a given bubble language $\Lang$.\label{algo:all}}

  \begin{algorithmic}[1]
\Statex
\Procedure{GenerateAll}{$n$}

\For {$d\gets 0$ {\bf to} $n$}  
	\State $w_1w_2\cdots w_n \gets 1^d0^{n-d}$
	\State \Call{GenBubble}{$d,n-d$}
\EndFor

\EndProcedure
  \end{algorithmic}
\end{algorithm}

{\em Remark: } It is even possible to give a cyclic Gray code for $\Lang(n)$, by giving the fixed-weight subsets listed first by the odd weights (increasing), followed by the even weights (decreasing). 

\medskip 

The oracle of Algorithm~\ref{algo:framework} applies a simple membership tester to compute the bubble upper bound for given $w\in \Lang$. However, we do not actually need a {\em general} membership tester, since all we want to know is which of the children of a node {\em already known to be in ${\cal L}$} are in ${\cal L}$; moreover, the membership tester is allowed to use other information, which it can build up iteratively while examining earlier nodes. In the next section, we will apply this method to the language of prefix normal words.

\section{A Gray Code for Prefix Normal Words} \label{sec:gray}

In this section, 
we prove that the set of prefix normal words $\LPN$ is a bubble language.   Then using the bubble framework and applying a basic quadratic-time membership tester, we show how to generate all  words in $\LPN(n,d)$ in Gray code order.  By concatenating the lists together for all weights in increasing order, we obtain an algorithm to list $\LPN(n)$ as a Gray code in $O(n^2)$ amortized time per word.  By then providing an enhanced membership tester for prefix normal words specific to the bubble framework, we further show how this Gray code can be generated in $O(\log^2 n)$ amortized time per word.

\begin{theorem} $\LPN$ is a bubble language.
\end{theorem}
\begin{proof}
Let $w$ be a prefix normal word containing an occurrence of $01$. Let $w'$ be the word obtained from $w$ by replacing the first occurrence of $01$ with $10$. Then $w=u01v$, $w'=u10v$ for some $u,v\in \Sigma^{*}$. Let $z$ be a substring of $w'$. We have to show that $|z|_{1}\le P(w',|z|)$.
 
Note that for any $k$, $P(w,k)\le P(w',k)$. In fact, $P(w',|u|+1) = P(w,|u|)+1$, and for every $k\neq |u|+1$, $P(w,k) = P(w',k)$. Now if $z$ is contained in $u$ or in $v$, then $z$ is a substring of $w$, and thus $|z|_1\le P(w,|z|) \le P(w',|z|)$. If $z=u'10v'$, with $u'$ suffix of $u$ and $v'$ prefix of $v$, then  $|z|_{1}=|u'01v'|_{1}\le P(w,|z|) \le P(w',|z|)$. If $z=0v'$, with $v'$ prefix of $v$, then $|z|_{1}<|1v'|_{1}$, and $1v'$ is a substring of $w$, thus $|z|_{1}\le P(w,|z|)\le P(w',|z|)$. 
Else $z=u'1$, with $u'$ suffix of $u$. We can assume that $u'$ is a proper suffix of $u$. Let $z'$ be the substring of $w'$ of the same length as $z$ and starting one position before $z$ (in other words, $z'$ is obtained by shifting $z$ to the left by one position). Since $u$ does not contain $01$ as a substring, we have $u=1^{n}0^{m}$ for some $n\ge 1, m\ge 0$. If $z'$ is a power of $0$'s, then $|z|_{1}=1$ and the claim holds. Else, $|z|_{1}=|z'|_{1}$, and $z'$ is a substring of $w$. Thus $|z|_{1}\le P(w,|z|) \le P(w',|z|)$. \hfill 
\end{proof}

Since there is a membership tester for prefix-normal words that runs in $O(n^2)$ time, e.g.\ as described in Algorithm~\ref{algo:member}, the aforementioned Gray codes for 
both $\LPN(n,d)$ and $\LPN(n)$ can be generated in $O(n^2)$ amortized time (Prop.~\ref{prop:bubble-gray}).  We show the computation tree $T^7_4$ in Fig.~\ref{fig:example1}, with prefix normal words in bold. The complete listing for $d=0,1,\ldots, 7$ is given in Table~\ref{table:pnw7}. 
%

\begin{algorithm}[hbt]
  \caption{Test if $w_1w_2\cdots w_n \in  \LPN$ in $O(n^2)$ time.\label{algo:member}}

  \begin{algorithmic}[1]
\Statex
\Function{Member}{ $\LPN, \  w_1w_2\cdots w_n$ }

\State $p_0 \gets 0$
\For {$i\gets 1$ {\bf to} $n$}  \ \ $p_{i} \gets p_{i-1} + w_i$
\EndFor
   
\For {$i\gets 2$ {\bf to} $n$} 
  \State $f \gets 0$
   \For {$j\gets i$ {\bf to} $n$} 
        \State $f \gets f + w_j$
        \If {$f > p_{j-i+1}$} \Return \sc{False}  \EndIf
   \EndFor
\EndFor
\State \Return \sc{True}

\EndFunction
  \end{algorithmic}
\end{algorithm}

\subsection{A More Efficient Approach}

Now we develop a more efficient membership tester for $\LPN$ that is specific to one required by an oracle for bubble languages.
In particular, membership tests are only made on strings of the form $w = 1^{s-1}0^{j}10^{t-j}\gamma$,  given that $1^s0^t\gamma \in \LPN$.

\begin{lemma}\label{lemma:isPNF}
Let $w = 1^s0^t\gamma$ be a word in $\LPN$ where $s \geq 0$, $t \geq 1$ and $\gamma \in 1\{0,1\}^* \cup \{\epsilon\}$.  
Let $w' = b_1b_2\cdots b_n =  \swap(w,s,s+j)$ for some $1 \leq j \leq t$.  
Then  $w'$ is {\bf not} in  $\LPN$ if and only if  either
\begin{enumerate}
\item $F(\gamma0^{s+t}, s+j-1) \geq s$, or 
\item  $|b_{s+j} b_{s+j+1} \cdots b_{2(s+j-1)}|_1 \geq s$.
\end{enumerate}
\end{lemma}

\begin{figure}
\begin{center}
\includegraphics[width=\textwidth]{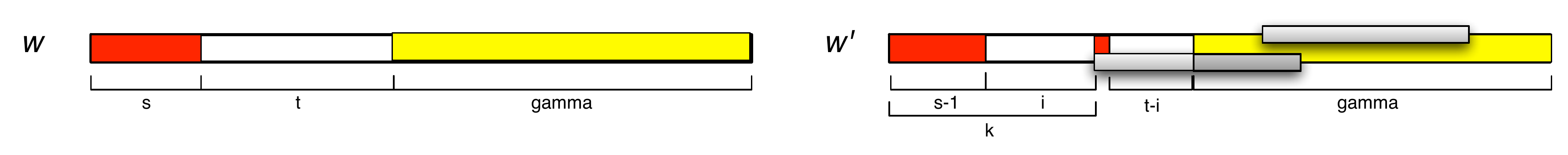}
\end{center}
\caption{Illustration of Lemma~\ref{lemma:isPNF}. On the right we have the two cases of a substring $u$ (in gray) of $w'$ which may violate the prefix normal property.\label{fig:isPNF}}
\end{figure}

\begin{proof}  The proof is illustrated in Fig.~\ref{fig:isPNF}. Note that  $w' = 1^{s-1}0^j10^{t-j}\gamma$. 

($\Leftarrow$) The prefix of $b_1b_2\cdots b_{s+j-1}$ of $w'$ has $s-1$ ones.  If $F(\gamma0^{s+t}, s+j-1) \geq s$, then it must be that $\gamma$ contains a substring of 
length $s+j-1$, or less,  with at least $s$ ones. Similarly, if  $|b_{s+j} b_{s+j+1} \cdots b_{2(s+j-1)}|_1 \geq s$ then $w'$ also contains a substring of length $s+j-1$ with at least $s$ ones.  Thus, if either $F(\gamma0^{s+t}, s+j-1) \geq s$ or $|b_{s+j} b_{s+j+1} \cdots b_{2(s+j-1)}|_1 \geq s$ then $w'$ is not in $\LPN$.

($\Rightarrow$) Assume $w'$ is not in $\LPN$. Then there is a shortest substring $u = b_ib_{i+1}\cdots b_{i+m-1}$ with length $m$ in $w'$ such that  $|u|_1 > P(w',m)$.   Clearly $i>1$ and since $m$ is minimal $b_i = b_{i+m-1} = 1$.  
Suppose $i+m-1 \leq s+j$.  Then $u$ can have at most $s-1$ ones since $i > 1$ and thus $|u|_1 \leq P(w',m)$, a contradiction.  Thus $i+m-1 >  s+j$.  
We now consider three cases for $i$.   If $i < s+j$ then since $b_i=1$, $2 \leq i < s$.  But since $b_1b_2\cdots b_{i-1} = 1^{i-1}$, this means $P(w',m) \geq |u|_1$.  Thus $i \geq s +j$.
Suppose $i=s+j$.  
If $m > s+j-1$  then the prefix of $w'$ of length $m$ overlaps with $u$, i.e.\ we can write $b_1b_2\cdots b_m = vv'$ and $u=v'u'$ for some non-empty $v'$ containing the swapped $1$. Since $|u|_1 > P(w',m)$, this implies that also $u'$ has more $1$s than the prefix of the same length, a contradiction to our choice of $u$.
Thus $m \leq  s+j-1$.  Since $w'$ starts with $1^{s-1}$ and $|u|_1 > P(w',m)$  it must be that $|u|_1 \geq s$.  By extending $u$ to have length $s+j-1$ we have $|b_{s+j} b_{s+j+1} \cdots b_{2(s+j-1)}|_1 \geq s$. 
Finally, suppose $i > s+j$.  Then because $b_i=1$,  $u$ is a substring of $\gamma$ and hence a substring of $w$.
Since $w \in \LPN$ we have $|u|_1 \leq P(w,m)$. 
Since $P(w,m) = P(w',m)$ for all $m<s$ and $m \geq s+j$, and  $|u|_1 > P(w',m)$, it must be that $s \leq m \leq s+j-1$. 
For each of these possible values for $m$, $P(w',m) = s-1$.  Thus $|u|_1 \geq s$ which means $F(\gamma, m) \geq s$.   Finally, since the length
of $\gamma0^{s+t}$ is at least $s+j-1$, we also have $F(\gamma0^{s+t}, s+j-1) \geq s$.  Considering all cases, we must either have  $F(\gamma0^{s+t}, s+j-1) \geq s$, or  $|b_{s+j} b_{s+j+1} \cdots b_{2(s+j-1)}|_1 \geq s$.
\end{proof}

Let $f_i$ denote the value $F(\gamma0^{s+t}, i)$.   By maintaining $f_1, f_2, \ldots, f_{s+t}$ as {\sc GenBubble} iterates through the prefix normal words,  we can apply the previous lemma to optimize a membership tester.  Pseudocode is given in Algorithm~\ref{algo:member2}.  This function requires the passing of the variables $s$ and $j$ from the function {\sc Oracle}, recalling that the current string $w_1w_2\cdots w_n$ is stored globally.

\begin{algorithm}[hbt]
  \caption{Membership testing for $\LPN$ specific to cool-lex framework in $O(s+t)$ time. \label{algo:member2}}

  \begin{algorithmic}[1]
\Statex
\Function{MemberPN}{$s,j$}

    \State $ones \gets 1$  \ \ \  \blue{  \Comment{first 1 accounted for by the (proposed) swap of a 1 to $w_{s+j}$}}
    \For {$i \gets s+j+1$ {\bf to} $2(s+j-1)$}   \If {$w_i  = 1$}  \ $ones \gets ones+1$ \EndIf  \EndFor

    \If { $ones \geq s$ {\bf or}  $f_{s+j-1} \geq  s$}   \Return \sc{False} \EndIf

    \State \Return \sc{True}

\EndFunction
  \end{algorithmic}
\end{algorithm}

The resulting membership tester clearly runs in $O(s+t)$ time since $j \leq t$.  In order to apply the oracle, we must maintain the data structure $f_1, f_2, \ldots, f_{s+t}$ as we proceed through the recursive generation algorithm. 
Since the length of the critical prefix $1^s0^t$ decreases as we go deeper in the computation tree, it is sufficient to update $f_1, f_2, \ldots, f_{s+i}$
as the bits in positions $s$ and $s+i$ get swapped.  Observe that this swap changes $\gamma$ by replacing the prefix $0^{i}$ with 1.  Thus,
to update the $f$ values we can simply scan the first $s+i$ bits (of the updated $\gamma$) as illustrated in {\sc UpdateF}($x$) of Algorithm~\ref{algo:update}, where $x$ corresponds to $s+i$.
This function should be inserted just before the recursive call is made in {\sc GenBubble}($s,t$) on line 9 of Algorithm~\ref{algo:framework}.  Since these values need to be restored \emph{after} the recursive call, we need to first save the initial values $f_1, f_2, \ldots, f_{s+i}$ in a temporary array so they can be restored.
A complete C implementation is given in Appendix A. 

\begin{algorithm}[h]
  \caption{Update required values for $f$ where $x = s+i$.\label{algo:update}}

  \begin{algorithmic}[1]
\Statex
\Procedure{UpdateF}{$x$}

\State  $ones \gets 0$ 	
\For{$k \gets x$ {\bf to}  $2x$}  
	\If{$w_k = 1$}    $ones \gets ones+1$  \EndIf
	\State $f_{k-x+1} \gets$ \Call{Max}{$f_{k-x+1}, ones$}
\EndFor

\EndProcedure
  \end{algorithmic}
\end{algorithm}

\subsection{Analysis}
\label{sec:analysis}

For each prefix normal word $w$ generated by {\sc GenerateAll}($n$), using the optimized membership tester for prefix normal words, the algorithm requires $O(cr(w))$ time to update $f_1, f_2, \ldots, f_{cr(w)}$.  Thus, the overall work done by the algorithm is proportional to $C(n) = \sum_{w \ \in  \ \LPN(n)}  cr(w)  = \pnw(n) \cdot O(\log^2 n), $ by Theorem~\ref{thm:expected-critical-prefix}. We summarize: 

\begin{theorem}
The set of words $\LPN(n)$, where $n > 1$, can be generated in amortized $O(\log^2 n)$ time per word.
\end{theorem}

\noindent
Prefix normal words are  the first interesting example of a bubble language for which no $O(1)$ amortized time generation algorithm is known.

\section{Membership Testing}
\label{sec:member}

The best membership testing algorithm uses the fact that a word is prefix normal if and only if it equals its prefix normal form (see Lemma~\ref{lemma:pnw_basics}). As mentioned before, the  most  efficient algorithm for computing the $F$-function of a word $w$, and thus its prefix normal form, is from~\cite{ChanL15} and runs in time $\Oh(n^{1.864})$. Here we present a simple two-phase membership tester which, even though $\Oh(n^2)$ in the worst-case, could outperform other algorithms in practice. 

Now consider the following two-phase approach.  Suppose there is an $O(n)$ test that rejects $X_n$ binary strings of length $n$ (Phase I).   Then, for the remaining $Y_n = 2^n{-}X_n$ strings,  apply the worst case $O(n^2)$ algorithm (Phase II).    On average this will lead to a membership algorithm that runs in time:
\[ \frac{c_1n \cdot X_n + c_2n^2 \cdot Y_n}{2^n},\] 
for some constants $c_1$ and $c_2$.
This expression will be less than $(c_1 + c_2)n $ if $Y_n \leq 2^n/n$, which implies an $O(n)$ average case tester.  Thus, when designing an $O(n)$ time rejection tester in Phase I, we 
are aiming to reject a number proportional to $2^n - 2^n/n$.  Without knowing exactly how many strings get rejected by a particular tester, we can focus  on the following ratio:
$$ ratio = \frac{n Y_n}{2^n}.$$
As $n$ grows, if this ratio is decreasing and bounded by a constant $c$, then $Y_n \leq c 2^n/n$, which implies an $O(n)$ average case time tester.

Applying this approach, we try the following trivial $O(n)$ test for Phase I:  a string will \emph{not} be prefix normal if the longest run of 1s is not a prefix.  Applying this test as the first phase, the resulting ratios for some increasing values of $n$ are given in Table~\ref{tab:ratios}(a).  Since the ratios are increasing as $n$ increases, we require a more advanced rejection test.

The next attempt uses a more compact \emph{run-length} representation for $w$.  Let $w$ be  represented by a series of $c$ blocks, which are maximal substrings of the form $1^*0^*$. Each block $B_i$ is composed of two integers $(s_i, t_i)$ representing the number of 1s and 0s respectively. For example, the string  
11100101011100110
can be represented by $B_1B_2B_3B_4B_5 = (3, 2)(1, 1)(1,1)(3,2)(2,1)$ where $c=5$.  Such a representation can easily be found in $O(n)$ time.  
A word $w$ will \emph{not} be prefix normal word if it contains a substring of the form $1^{i}0^j1^k$ such that
$i+j+k \leq s_1 + t_1$ and $i+k > s_1$ (the substring is not longer, yet has more 1s than the critical prefix). Thus, a word is not prefix normal if for some $2 \leq i \leq c$:
\[s_{i-1} + t_{i-1} + s_i \leq  s_1 + t_1     \text{  \  \ and  \  \ }  s_{i-1} + s_{i} > s_1.\]

By applying this additional test in our first phase, we obtain Algorithm~\ref{algo:member3}.  The ratios that result from this algorithm are given in Table~\ref{tab:ratios}(b).  Similar decreasing ratios also occur for odd $n$.  \\

\begin{table}[h]
\begin{center}
\begin{tabular}{c  | c}
$n$  	&   $ratio$ \\  \hline
10    &   2.500   \\
12    &   2.561    \\
14    &   2.602   \\
16    &   2.631   \\
18    &   2.656  \\
20    &   2.675   \\
22    &   2.693  \\
24    &   2.708  \\
\multicolumn{2}{c }{ }   \\
\multicolumn{2}{c }{(a)}  
\end{tabular}
\ \ \ \ \ \ \ \ \ \ \ 
\ \ \ \ \ \ \ \ \ \ \ 
\ \ \ \ \ \ \ \ \ \ \ 
\begin{tabular}{c  |  c}
$n$  	&   $ratio$ \\  \hline
10    &   2.168  \\
12    &   2.142    \\
14    &   2.121   \\
16    &   2.106   \\
18    &   2.093  \\
20    &   2.083  \\
22    &   2.075  \\
24    &   2.067  \\
\multicolumn{2}{c }{ }   \\
\multicolumn{2}{c }{(b)}  
\end{tabular}
  \caption{(a) Ratios from the trivial rejection test.  \ (b) Ratios by adding secondary rejection test.}
  \label{tab:ratios}

\end{center}

\end{table}

 Since the ratios achieved by the combination of the two tests are decreasing as $n$ increases (see Table~\ref{tab:ratios}(b)), we make the following conjecture: 

\begin{conjecture}
The membership tester {\sc MemberPNF}($w$) for $\LPN$ runs in  $O(n)$-time on average, where the average is taken over all words of length $n$. 
\end{conjecture}

We note that one can conceive of several other (perhaps more advanced) rejection tests that run in $O(n)$ time, however, these two were sufficient to obtain our desired experimental results.

\begin{algorithm}[t]
  \caption{Membership tester: returns whether or not $w=w_1w_2\cdots w_n \in  \LPN$.\label{algo:member3}}

  \begin{algorithmic}[1]
\Statex
\Function{MemberPNF}{   $w$ }

\State$(s_1,t_1)(s_2,t_2)\cdots (s_c,t_c)  \gets$ the run-length block encoding of $w$

\State  \blue{$\triangleright$  \ \em{Phase I: linear time rejection tests}}
\For {$i\gets 2$ {\bf to} $c$}  
		\If{$s_i > s_1$}  \Return {\sc False}  \EndIf
		\If{$s_{i-1} + t_{i-1} + s_i \leq  s_1 + t_1$    {\bf   and   } $ s_{i-1} + s_{i} > s_1$}  \Return {\sc False}  \EndIf
\EndFor

\State  \blue{$\triangleright$  \ \em{Phase II: call $O(n^2)$ membership tester}}   

\State \Return \Call{Member}{ $\LPN, \  w$ }

\EndFunction
  \end{algorithmic}
\end{algorithm}

\section{Conclusion and Open Problems}\label{sec:conclusion}

The main result of this paper is a generation algorithm for prefix normal words, which is shown to run in amortized $\Oh(\log^2 n)$ time per word, as opposed to the hitherto best $\Oh(n)$ time per word algorithm. Our algorithm is based on the fact that prefix normal words form a bubble language, thus the general framework for bubble languages can be applied, and the algorithm outputs the language as a Gray code. 

We further presented a novel view of bubble languages, in terms of subtrees of the computation tree of a generating algorithm for all binary strings. We hope that this view will aid readers to apply the bubble framework to other binary languages. 
Finally, we gave a membership tester for prefix normal words, which we conjecture to run in average linear time over all binary words. 

We conclude with the following open problems on prefix normal words:
\begin{enumerate}
\item Given a prefix normal word $w$, efficiently list all words with prefix normal form $w$ (i.e., its equivalence class). The maximum size of an equivalence class is listed in the OEIS as sequence A238110~\cite{oeis}. Note that in the recent article~\cite{BG19}, the authors prove that the maximum equivalence class size is asymptotically $2^{n-O(\sqrt{n \log n})}$. 
\item Derive a closed form enumeration formula for the number $\pnw(n)$ of prefix normal words of length $n$, or a generating function for $\pnw(n)$. 
\item Develop an algorithm to exhaustively list all prefix normal words in constant amortized time per word. 
\item Develop a general membership tester for prefix normal words which runs in $o(n^{1.864})$ time in the worst case. 
\end{enumerate}

\subsection*{Acknowledgements}
We thank Frank Ruskey for useful discussions. We further thank the organizers of the Dagstuhl Seminar no.\ 18281~\cite{BarbayDR18281}, which took place in July 2018, 
and which gave two of the authors an opportunity to collaborate on prefix normal words. Gabriele Fici is supported by MIUR project PRIN 2017K7XPAN ``Algorithms, Data Structures and Combinatorics for Machine Learning''.



\small

\newpage
\noindent
\Large
{\bf Appendix A: C code}  
\normalsize

\scriptsize
\begin{code}
//-----------------------------------------------
// COOL-LEX GRAY CODE for Prefix Normal Words
// OEIS:  http://oeis.org/A194850
//-----------------------------------------------
#include <stdio.h>
#define MAX(a,b) ((a > b) ? a : b)

int a[100], F[100], N, NO_OUTPUT=0, COLEX=0;
long long int total = 0;

//---------------------------------
void Visit() {
int i;

	if (!NO_OUTPUT) {
       		for (i=1; i<=N; i++) {
			if (a[i] == 0) printf("1");
			else printf("0");
		}
		printf("\n");    
	}
	total++;
}
//---------------------------------
void Swap(int i, int j) {
int tmp;
    
	tmp = a[i];    a[i] = a[j];     a[j] = tmp;
}
//---------------------------------
int Member_PNF(int s, int j) {
int i, ones=1;

	for (i=s+j+1; i<=2*(s+j-1); i++) if (a[i] == 1) ones++;
	if (ones >= s || F[s+j-1] >= s) return 0;
	return 1;
}
//---------------------------------
int Oracle_PNF(int s, int t) {
int j=1;	

	while (j <= t && Member_PNF(s,j)) j++;
	return j-1;
}
//---------------------------------
int UpdateF(int x) {
int j, ones=0;

	for (j=x; j<=2*x; j++) { 
		if (a[j] == 1) ones++;
		F[j-x+1] = MAX(F[j-x+1],ones);
	}
}
//----------------------------
// COOL LEX GRAY CODE or COLEX
//----------------------------
void Gen(int s, int t) {
int i, j, k, G[100];

	if (COLEX) Visit();	
	if (s > 0 && t > 0) {
		j = Oracle_PNF(s,t);   
		for (i=1; i<=j; i++) {
			Swap(s,s+i);
			for (k=s+i; k<=2*(s+i); k++) G[k-(s+i)+1] = F[k-(s+i)+1];
			
			UpdateF(s+i);
			Gen(s-1,i);
			
			for (k=s+i; k<=2*(s+i); k++) F[k-(s+i)+1] = G[k-(s+i)+1];
			Swap(s,s+i);
		}    
	} 	
	if (!COLEX) Visit();
}   
//---------------------------------
int main( ) {
int i, j, output, D;

	//-------
	// INPUT
	//-------
	printf("\nSELECT output [1]Cool-lex Gray code [2]Co-lex [3]Just counts: ");   
	scanf("%d", &output); 
	
	if (output == 2) COLEX = 1;
	if (output == 3) NO_OUTPUT = 1;
		
	printf("ENTER length n: ");    scanf("%d", &N);
	printf("ENTER # of a's (or -1 for all PN words): ");    scanf("%d", &D); 
	printf("\n");
	
	for (i=1; i<=N; i++) F[i] = 0;			        
	
	//-----------
	//GENERATION
	//-----------
    	if (D == -1) {
        	for (j=0; j<=N; j++)  {
			for (i=1; i<=j; i++)   a[i] = 1;
			for (i=j+1; i<=2*N; i++) a[i] = 0;	// PAD WITH N 1s
			Gen(j,N-j); 
        	}
    	}
    	else {		
		for (i=1; i<=D; i++)     a[i] = 1;
		for (i=D+1; i<=2*N; i++) a[i] = 0;		// PAD WITH N 1s
		Gen(D,N-D);
	}
	
	printf("Total = %lld\n", total);
}
\end{code}

\end{document}